\documentclass[conference]{IEEEtran}
\IEEEoverridecommandlockouts
\usepackage{cite}
\usepackage{amsmath,amssymb,amsfonts,amsthm}
\usepackage{algpseudocode}
\usepackage{graphicx}
\usepackage{textcomp}
\usepackage{xcolor}
\usepackage{hyperref}
\usepackage{fancyhdr}

\def\BibTeX{{\rm B\kern-.05em{\sc i\kern-.025em b}\kern-.08em
    T\kern-.1667em\lower.7ex\hbox{E}\kern-.125emX}}

\newcommand{\E}{\mathbb{E}}
\newcommand{\PR}{\mathbb{P}}

\newcommand{\calI}{\mathcal{I}}

\newcommand{\calA}{\mathcal{A}}

\newcommand{\calF}{\mathcal{F}}

\newcommand{\calZ}{\mathcal{Z}}

\newcommand{\calB}{\mathcal{B}}

\newcommand{\bfz}{\mathbf{z}}
\newcommand{\bfI}{\mathbf{I}}
\newcommand{\bfm}{\mathbf{m}}
\newcommand{\bfx}{\mathbf{x}}
\newcommand{\bfX}{\mathbf{X}}
\newcommand{\set}[1]{\left\{#1\right\}}

\newtheorem{thm}{Theorem}

\newtheorem{pro}{Proposition}

\usepackage{algorithm}
\usepackage{algpseudocode}

\begin{document}

\title{Optimally revealing bits for rejection sampling}

\author{\IEEEauthorblockN{Louis-Roy Langevin}
\IEEEauthorblockA{\textit{Dept. of Mathematics and Statistics} \\
\textit{McGill University}\\
Montréal, Canada \\
louis-roy.langevin@mail.mcgill.ca}
\and
\IEEEauthorblockN{Alex Waese-Perlman}
\IEEEauthorblockA{\textit{Dept. of Mathematics and Statistics} \\
\textit{McGill University}\\
Montréal, Canada \\
alex.waese-perlman@mail.mcgill.ca}
}

\maketitle
\thispagestyle{fancy}

\begin{abstract}
Rejection sampling is a popular method used to generate numbers that follow some given distribution. We study the use of this method to generate random numbers in the unit interval from increasing probability density functions. We focus on the problem of sampling from $n$ correlated random variables from a joint distribution whose marginal distributions are all increasing. We show that, in the worst case, the expected number of random bits required to accept or reject a sample grows at least linearly and at most quadratically with $n$.
\end{abstract}


\section{Introduction}
Understanding how much information we need to learn about a number to answer questions about it is a fundamental problem in computer science. Imagine you pick a random number $X$ uniformly between 0 and 1, but you don’t know its exact value. Instead, you can reveal the digits of its binary representation one at a time. For example, if $X = 1/3$, its binary form is $0.010101...$, and each digit gives a clue about where $X$ lies.

Each digit of $X$ is like a fair coin toss — either 0 or 1 with equal chance. If you want to know whether $X$ is less than some fixed number $x$, on average you need to look only at the two first bits of $X$ to find out.

Now, consider two random numbers, $X_1$ and $X_2$, chosen independently and uniformly between 0 and 1. Our goal is to decide if $X_1 < X_2$ by revealing bits from either number, one by one, in any order we choose. The best way is to alternate between revealing bits from $X_1$ and $X_2$. With this strategy, we expect to reveal about 4 bits in total before finding out the answer.

We can generalize this question: given a known increasing function $f$, how many bits from $X_1$ and $X_2$ must we reveal to decide if $f(X_1) < X_2$ on average? This question is central to our work and has important applications in a classic technique called rejection sampling.

\subsection{Rejection sampling}

The challenge of generating random numbers has been addressed by many, including Devroye \cite{dev}, and remains an active topic of interest in computer science. Various methods based on quantum computing \cite{ce} and spin electronics \cite{dcvrkh} have been recently developed and have many applications in statistics, physics, finance, etc.

Rejection sampling is a popular method for generating random samples from given probability distributions. First introduced by John von Neumann in 1951 \cite{von}, it is widely used in many fields including machine learning \cite{kfjbb}\cite{xyxpwsljzxd}, cryptography \cite{ar}, and numerical analysis.

Here is how it works: given a function $f$ that corresponds to a probability density for the points in $[0,1)^n$, we repeatedly pick random points $(X_1,\ldots,X_n)$ and $X_{n+1}$ uniformly and independently in $[0,1)^n$ and $[0,1)$, respectively. We keep doing this until $f(\bfX)$ is less than $X_{n+1}$. The point $\bfX$ we stop at then follows a distribution proportional to $f$.

In the following picture, we ran rejection sampling many times with $n=1$. The points in blue are the ones that have been accepted and outputted because they lie under the function $f$. Notice that the values of $\bfX$ for which $f$ is higher are outputted more often than those where $f$ has low value, as highlighted by the two blue intervals.

\begin{figure}[h]
    \centering
    \includegraphics[width=200pt]{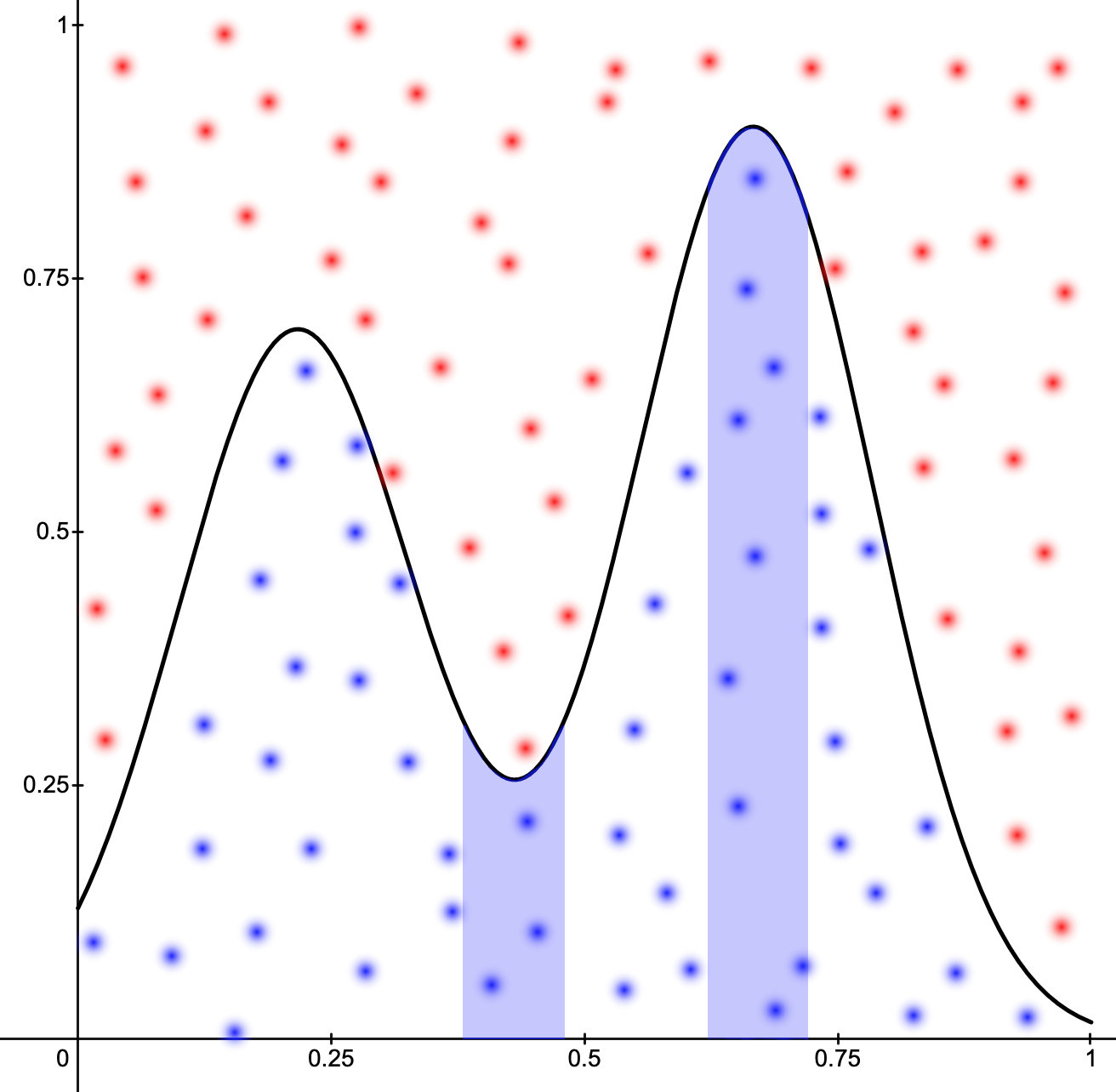}
    \caption{Visual representation of rejection sampling with $n=1$}
\end{figure}

When implementing rejection sampling on a computer, one challenge is to minimize how many bits of these numbers we need to reveal before knowing if $f(\bfX) < X_{n+1}$. This is the problem we focus on in this paper.

We assume we can compute $f(\bfx)$ instantly for any input $f(\bfx)$, so the challenge is to minimize queries to the bits of the $X_j$'s.

We focus on the case where $f$ is increasing: if every coordinate of $\bfx'$ is less than or equal to the corresponding coordinate of $\bfx$, then $f(\bfx') \le f(\bfx)$. This property lets us use just two points to decide if $f(\bfX)$ is bigger or smaller than $X_{n+1}$, which makes the algorithm much more realistic to implement.

\subsection{The setup}

Given a number $x \in [0,1)$, let $x^{(i)} \in \set{0,1}$ be the $i$'th bit of its binary expansion, i.e. $x = 0.x^{(1)}x^{(2)}x^{(3)}\ldots$. We assume that $(x^{(i)})_{i \ge 1}$ does not end with a tail of repeating 1's. This way, we have a bijection between $[0,1)$ and all such binary sequences.  Finally, denote $\calF_n$ to be the set of all increasing functions from $[0,1]^n$ to $[0,1)$.

To represent the algorithms we are studying, we will first consider deterministic inputs rather than random variables. We introduce the alphabet $\set{\bar{x}_1,\ldots,\bar{x}_{n+1}}$, where $\bar{x}_1,\ldots,\bar{x}_{n+1}$ are symbols that we call \textit{cuts}. An algorithm takes in a point $(\bfx,x_{n+1}) \in [0,1)^{n+1}$ and an increasing function $f \in \calF_n$, and outputs a string in $\set{\bar{x}_1,\ldots,\bar{x}_{n+1}}^*$ (possibly of infinite length) that represents the exact sequence of bits that have been revealed. For example, if the output of the algorithm is $\bar{x}_1\bar{x}_2\bar{x}_1\bar{x}_{n+1}\bar{x}_1$, it means that the bits $x^{(1)}_1,x^{(1)}_2,x^{(2)}_1,x^{(n+1)}_1,x^{(3)}_1$ have been revealed in this exact order. 

Such an algorithm is thus a function
\begin{align*}
    A : \calF_n \times [0,1)^{n+1} \to \set{\bar{x}_1,\ldots,\bar{x}_{n+1}}^*\,,
\end{align*}
with certain restrictions. Let $p_1,\ldots,p_{n+1} \ge 1$ be integers and $f \in \calF_n$. If $(\bfx,x_{n+1}),(\bfx',x'_{n+1}) \in [0,1)^n$ are such that $x^{(i)}_j = x'^{(i)}_j$ for all $1 \le i \le p_j, 1 \le j \le n+1$, then if $s$ is a prefix of $A(f,\bfx,x_{n+1})$ that contains $\bar{x}^{(j)}$ at most $p_j$ times for all $1 \le j \le n+1$, then any prefixes of lengths $|s|+1$ or less of $A(f,\bfx,x_{n+1})$ and $A(f,\bfx',x'_{n+1})$ must be identical. This restriction ensures that the values of the bits that are not revealed cannot have an impact on which bit is queried next.

The last restriction is that given $(\bfx,x_{n+1}) \in [0,1)^{n+1}$ and $f \in \calF_n$, the string of revealed bits outputted by these algorithms must provide sufficient information to determine whether $f(\bfx) < x_{n+1}$ or not. Let $s$ be a string of cuts, $k_i$ be the number of $x^{(i)}$'s in $s$, and $x^*_i = 0.x_i^{(1)}x_i^{(2)}\ldots x^{(k_i)}$, i.e. $x^*_i$ is the number consisting of the first $k_i$ bits of $x_i$. After revealing its first $k_i$ bits, $x_i$ could take any value in $[x_i^*,x_i^*+2^{-k_i})$. We thus have that $f(\bfx) < x_{n+1}$ almost surely\footnote{Here "almost surely" means that the set of values of $(\bfx,x_{n+1})$ that makes the equation false has measure $0$.} if and only if $f(x^*_1+2^{-k_1},\ldots,x^*_n+2^{-k_n}) < x_{n+1}^*$, and $f(\bfx) > x_{n+1}$ almost surely if and only if $f(\bfx^*) > x_{n+1}^* + 2^{-k_{n+1}}$. Note that $f(\bfx) \ne x_{n+1}$ almost surely.

Let us call $S = [\bfx^*_1,\bfx^*_1+2^{-k_1})\times\ldots\times[\bfx^*_{n+1},\bfx^*_{n+1}+2^{-k_{n+1}})$ the feasible set of $A$ after string $s$, i.e. the set of possible values of $(\bfx,x_{n+1})$. If $s$ does not satisfy any of the two halting conditions mentioned in the previous paragraph, then the algorithm cannot halt after outputting $s$. In this case, we say that $f$ \textit{crosses} $S$, because this means that $S\cap\set{(\bfx,x_{n+1}) : f(\bfx) < x_{n+1}}$ and $S\cap\set{(\bfx,x_{n+1}) : f(\bfx) > x_{n+1}}$ both have positive measure.

\subsection{The problem and our results}

The goal of this paper is to determine which algorithms query the lowest total number of bits from $X_1,\ldots,X_{n+1}$ in the worst case, on average. Formally, letting $A(f) = A(f,X_1,\ldots,X_{n+1})$, we want to study the value and the asymptotic behavior of
\[
\calB(n) = \inf_{A \in \calA_n} \sup_{f \in \calF_n} \E|A(f)|\,
\]
and find good algorithms that reach this value for all $n \ge 1$. We also aim to find what functions maximize the expected number of bits these algorithms reveal, i.e. the worst case functions.

In section \ref{ub}, we find a general upper bound on the value of $\calB(n)$. We prove in Theorem \ref{ub:thm1} that this bound grows at most quadratically with $n$, i.e. that $\calB(n)$ is $O(n^2)$.  In Section \ref{lb}, we show that $B(n) \ge n+1$, meaning that $B(n)$ is $\Omega(n)$.

\subsection{Preliminary results}

Throughout this paper, we will use the fact that if $N \ge 0$ is an integer valued random variable, then
\begin{equation}
\label{eq:exp}
    \E[N] = \sum_{\ell \ge 0} \PR(N > \ell).
\end{equation}
This can be seen after writing $N = \sum_{\ell \ge 0} \mathbf{1}_{N > \ell}$ and using linearity of expectation on the right hand side.

Next, let $x \ge 0$ be any real number and $n \ge 1$ be an integer. Then,
\begin{equation}
\label{eq:ineq}
    1 - (1 - x)^n \le nx
\end{equation}
with equality if and only if $x = 0$. To see this, take the derivatives on both sides to obtain $n(1-x)^{n-1}$ and $n$, respectively. For every $x > 0$, the derivative of the right hand side is strictly bigger, which gives us the inequality.

\section{An upper bound for $\calB(n)$}\label{ub}

First, we bound $\calB(n)$ above by constructing the naive alternating algorithm $ALT_n$ defined as follows. $ALT_n(f,\bfx,x_{n+1})$ is equal to the shortest possible prefix of $(\bar{x}_1\bar{x}_2\ldots\bar{x}_{n+1})^*$ that allows the algorithm to halt (of infinite length if and only if $f(\bfx) = x_{n+1}$). Clearly the two conditions for $ALT_n$ to be in $\calA_n$ are satisfied. We start the analysis of the efficiency of $ALT_n$ with the following preliminary result.

\begin{pro}
\label{ub:pro1}
    Let $M_1,\ldots,M_{n+1} \ge 1$ be integers and let $f : [0,M_1)\times[0,M_n) \to [0,M_
    {n+1})$ be an increasing function. For $\bfz = (z_1,\ldots,z_{n+1}) \in \set{1,2,\ldots}^{n+1}$ with $1 \le z_j \le M_j$, let $\bfI_\bfz = [z_1-1,z_1)\times\ldots\times[z_{n+1}-1,z_{n+1})$. Then, the number of such $\bfI_\bfz$'s the function $f$ crosses is at most $\prod_{j=1}^{n+1} M_j - \prod_{j=1}^{n+1} (M_j - 1)$.
\end{pro}

\begin{proof}
    Let $\calZ$ be the set of all $\bfI_{\bfz^*}$'s such that $z^*_{n+1} = M_{n+1}$ or $z^*_i = 1$ for some $1 \le i \le n$. For each $\bfI_{\bfz^*} \in \calZ$, let $\calI_{\bfz^*}$ be the set of all $\bfI_\bfz$ such that there exists an $m \ge 0$ with $z_{n+1} = z^*_{n+1}-m$ and $z_i = z^*_i+m$ for all $1 \le i \le n$.

    We show that $f$ cannot cross more than one rectangle in any $\calI_{\bfz^*}$. Let $\bfI_\bfz,\bfI_{\bfz'} \in \calI_{\bfz^*}$ and assume that $z_{n+1} \le z'_{n+1}-1$. If $f$ crosses $\bfI_\bfz$, then $f(z_1-1,\ldots,z_n-1) < z_{n+1}$, which implies $f$ does not cross $\bfI_{\bfz'}$ since $z'_i \le z_i - 1$ for all $1 \le i \le n$, and thus $f(z'_1,\ldots,z'_n) < z'_{n+1}-1$ since $f$ is increasing.

    Finally, notice that the total number of $\bfI_\bfz$'s is $\prod_{j=1}^{n+1} M_j$, and the number of $\bfI_\bfz$'s that are not in $\calZ$ is $\prod_{j=1}^{n+1} (M_j-1)$. Furthermore, every $\bfI_\bfz$ is contained in some $\calI_{\bfz^*}$, and since at most one $\bfI_\bfz$ per $\calI_{\bfz^*}$ can be crossed by $f$, this implies that at most $\prod_{j=1}^{n+1} M_j - \prod_{j=1}^{n+1} (M_j-1)$ of the $\bfI_\bfz$'s can be crossed by $f$ in total, i.e. at most one for each $\calI_{\bfz^*}$.
\end{proof}

\begin{thm}
\label{ub:thm1}
    Let $f : [0,1)^n \to [0,1)$ be any increasing function. Then,
    \begin{align*}
        \E\left|ALT_n(f)\right| \le \sum_{k \ge 0} \sum_{i=0}^n \left( 1  - \frac{(2^{k+1}-1)^i(2^k-1)^{n+1-i}}{2^{k(n+1)+i}} \right)\,.
    \end{align*}
    In particular, $\calB(n) = O(n^2)$.
\end{thm}

\begin{proof}
    Let $0 \le i \le n$ and $k \ge 0$ be integers. By observing the algorithm, we see that if $ALT_n$ has not halted after $k(n+1) + i$ queries, then its current string contains $k+1$ occurrences of each of $\bar{x}_1,\ldots,\bar{x}_i$, and $k$ occurrences of each of $\bar{x}_{i+1},\ldots,\bar{x}_n)$.
    
    Let $1 \le m_j \le 2^{k+1}$ be integers for $1 \le j \le i$ and let $1 \le m_j \le 2^k$ be integers for $i < j \le n+1$. Then, let $\bfI_\bfm=$\\ $\left[2^{-k-1}(m_1-1),2^{-k-1}m_1\right)\times\ldots\times\left[2^{-k-1}(m_{n+1}-1),2^{-k-1}m_{n+1}\right)$. Since $X_1,\ldots,X_{n+1}$ are uniform on $[0,1)$, every $\bfI_\bfm$ is equally likely to contain $(X_1,\ldots,X_{n+1})$. Applying Proposition \ref{ub:pro1} on a rescaling of $[0,1)^{n+1}$ gives us that at most $2^{k(n+1)+i} - (2^{k+1}-1)^i(2^k-1)^{n+1-i}$ of the $\bfI_\bfm$ can be crossed by $f$. Since the feasible set of $ALT_n$ must be equal to some $\bfI_\bfm$ after $k(n+1)+i$ cuts, this means the probability that $f$ still crosses the feasible set after $k(n+1)+i$ cuts is at most $[2^{k(n+1)+i} - (2^{k+1}-1)^i(2^k-1)^{n+1-i}]\cdot\PR(\bfI_\bfm)$. Since $\PR(\bfI_\bfm) = 2^{-k(n+1)-i}$, this means that
    \[
        \PR(|ALT_n(f)| > k(n+1) + i) = 1 - \frac{(2^{k+1}-1)^i(2^k-1)^{n+1-i}}{2^{k(n+1)+i}},
    \]
    thus equation \ref{eq:exp} gives us the upper bound we are looking for.

    To see that $\calB(n) = O(n^2)$, note that since $\PR(|ALT_n(f)| > \ell)$ decreases when $\ell$ increases, then we can naively bound the right hand side of this Theorem's inequality with
    \begin{equation}
    \label{eq:slack1}
        \sum_{k \ge 0} \sum_{i=0}^n \left(1 - \frac{(2^k-1)^{n+1}}{2^{k(n+1)}}\right).
    \end{equation}
    The terms inside of the double sum can be written as $1 - (1-2^{-k})^{n+1}$. We can now apply equation \ref{eq:ineq} to bound this above with $(n+1)2^{-k}$. In the end, this gives us an upper bound of
    \begin{equation}
    \label{eq:slack2}
        \sum_{k \ge 0} (n+1)^22^{-k} = 2(n+1)^2.
    \end{equation}
\end{proof}

When we plot original theorem's equation in Desmos and compare it with $2(n+1)^2$, we notice a clear gap between them. This could be due either to the bound in equation \ref{eq:slack1} or the bound in equation \ref{eq:slack2}. After plotting the three, it looks like equation \ref{eq:slack1} is very tight and may only alter the upper bound by a constant factor. However, equation \ref{eq:slack2} seems to make the bound go from a function that seems empirically close to $\Theta(n^{1.175})$ to a quadratic one.

Here are the plots we are referring to. In red, we have the first Theorem's inequality. In green, we have the bound from equation \ref{eq:slack1}. In purple, we have the bound from equation \ref{eq:slack2}. Both the green and the red functions look like they follow asymptotic behaviors of order $\Theta(n^{1.175})$. This can be seen by changing the value of $b$ to match the function.

\begin{figure}[h]
    \centering
    \includegraphics[width=250pt]{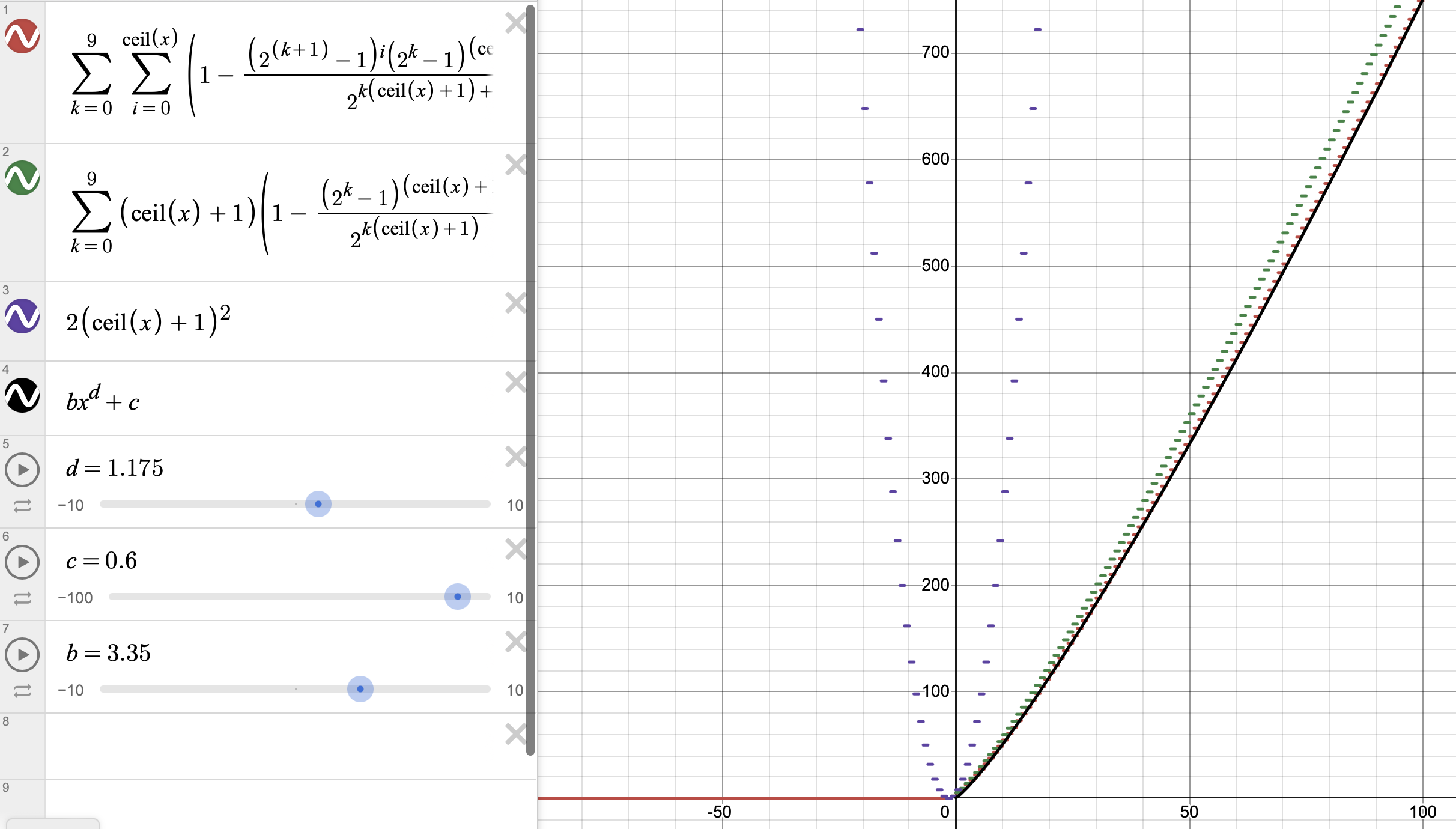}
    \caption{Graphical representation of the different upper bounds for $\calB(n)$.}
\end{figure}

This leads us to believe that the upper bound can be drastically improved if we find an alternative bound to equation \ref{eq:slack2}.

\section{A lower bound for $\calB(n)$}\label{lb}

In this section, we now prove that $\calB(n) \ge n+1$, letting us conclude that $\calB(n) = \Omega(n)$. We denote $h_n : [0,1)^n \to [0,1)$ to be the function in $\calF_n$ defined as $h_n(\bfx) = \min\left(1-2^{-(n+1)},(2^{n+1}-1)\min_{1 \le i \le n} x_i\right)$. Plots of $h_n$ are shown in the following figure for $n=1$ and $n=2$:

\begin{figure}[h]
    \centering
    \includegraphics[width=200pt]{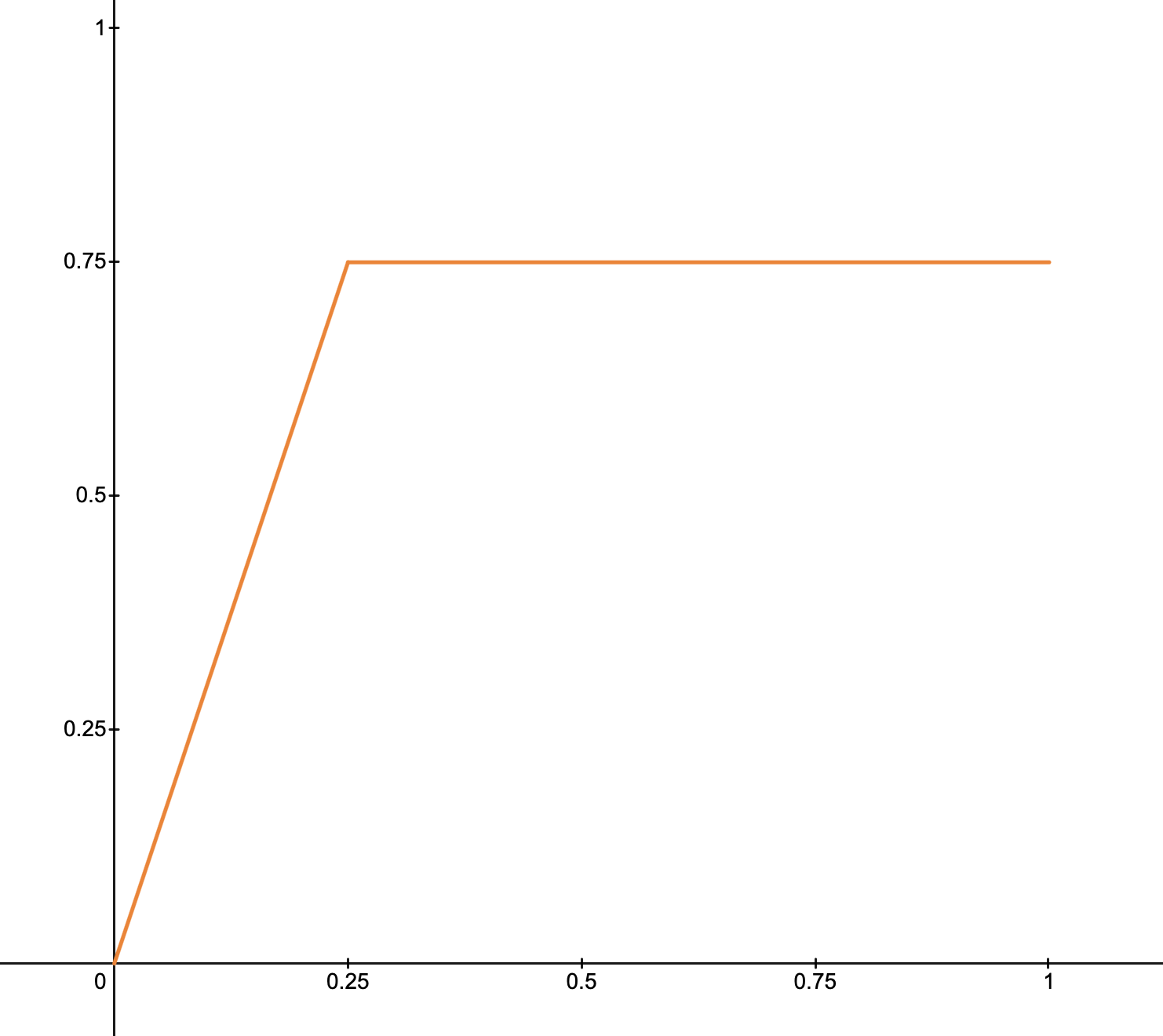}
    \caption{Graphical representation of $h_1$.}
\end{figure}
\begin{figure}[h]
    \centering
    \includegraphics[width=200pt]{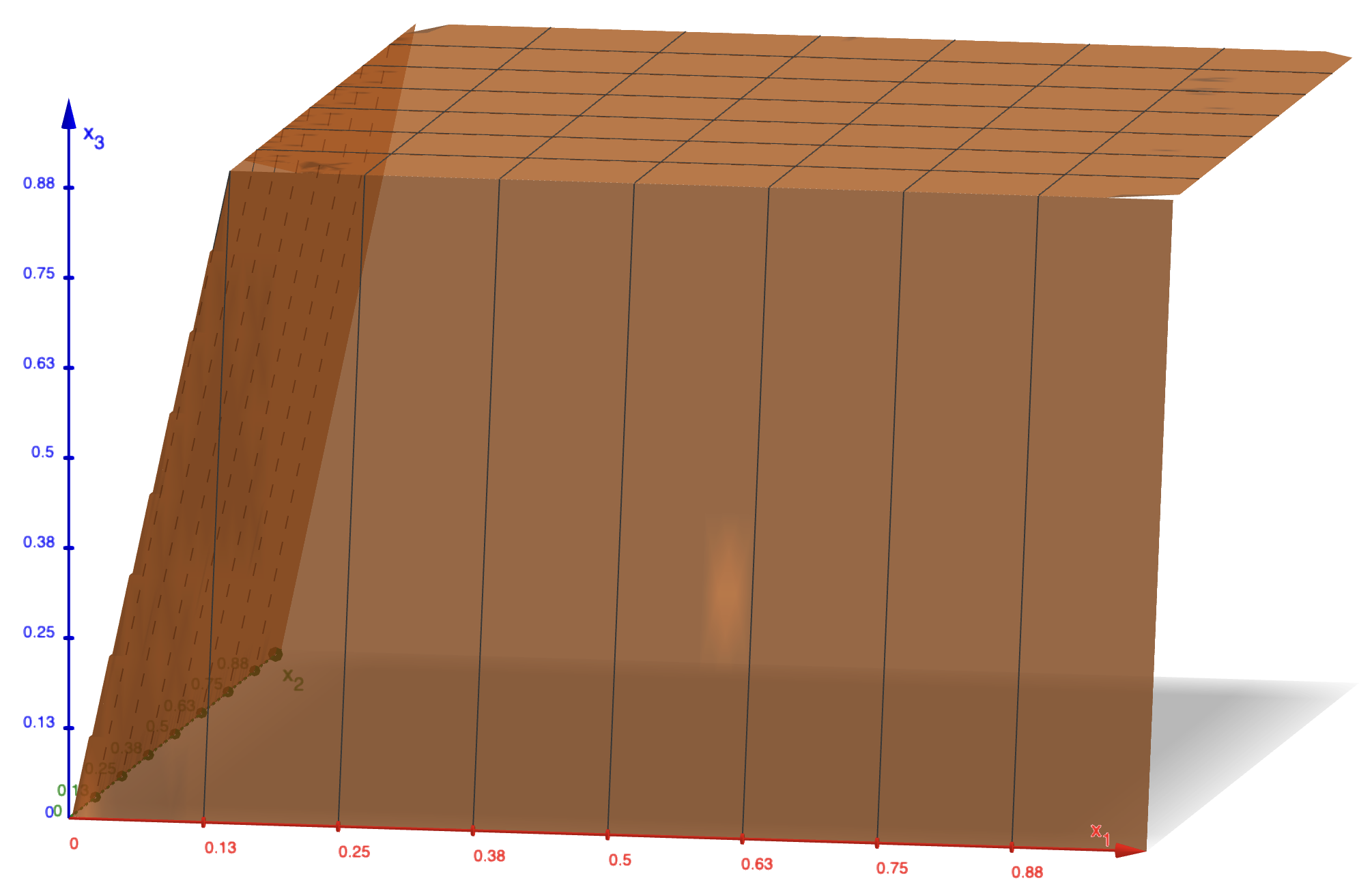}
    \caption{Graphical representation of $h_2$.}
\end{figure}

We use $h_n$ to prove our lower bound in the following Theorem.

\begin{thm}
    For any algorithm $A \in \calA_n$ and any $(\bfx,x_{n+1}) \in [0,1)^{n+1}$, $|A(h_n,\bfx,x_{n+1})| \ge n+1$. In particular, $\calB(n) \ge n+1$.
\end{thm}
\begin{proof}
    Let $A \in \calA_n$ be any algorithm and let $s$ be a prefix of length less than $n$ of $A(f,\bfx,x_{n+1})$. By the pigeonhole principle, there must be one of $\bar{x}^{(1)},\ldots,\bar{x}^{(n+1)}$ that isn't included in $s$. Let $S$ be the current feasible set of $A$. If $\bar{x}^{(n+1)} \notin s$, then $S$ contains points with $x_{n+1}$-value 0 and points with $x_{n+1}$-value bigger than $1-2^{-(n+1)}$. Since $0 < h_n \le 1 - 2^{-(n+1)}$ almost everywhere, this implies that $f$ crosses $S$, meaning that $A$ cannot halt.

    Now, if $\bar{x}^{(n+1)}$ is included in $s$, then choose $1 \le i \le n$ such that $\bar{x}^{(i)}$ is not included in $s$. This means that $S$ contains a point $(\bfx,x_{n+1})$ with $x_i = 0$ and $x_{n+1} > 0$, thus $h_n(\bfx) < x_{n+1}$ since $h_n(\bfx) = 0$ by observation. Furthermore, since $s$ contains at most $n$ cuts, then $S$ must contain a point $(\bfx',x_{n+1}')$ such that $x'_i > 2^{-(n+1)}$ for all $1 \le i \le n$ and $x_{n+1}' < 1 - 2^{-(n+1)}$. Thus, $h_n(\bfx') > x_{n+1}$ since $h_n(\bfx') = 1 - 2^{-(n+1)}$ by observation. The existence of $\bfx$ and $\bfx'$ implies that $h_n$ strictly crosses $S$, and so that $A$ cannot halt. Since $s$ is an arbitrary string of length at most $n$, this implies that $A$ needs at least $n+1$ cuts to halt, i.e. $|A(h_n,\bfx,x_{n+1})| \ge n+1$ for any $(\bfx,x_{n+1}) \in [0,1)^{n+1}$, proving the desired result.

    We can use this deterministic result and replace $\bfx$ and $x_{n+1}$ by $\bfX$ and $X_{n+1}$, respectively, to get that $A(h_n) \ge n+1$. Taking the supremum over any function $f \in \calF_n$ gives that $\sup_{f \in \calF_n} A(f) \ge n+1$, and since this inequality is true for any algorithm $A \in \calA_n$, we can conclude that $\calB(n) \ge n+1$.
\end{proof}

\section{Conclusion}

To generate random numbers that follow a given distribution on $[0,1)^n$, we use rejection sampling driven by sequences of independent random bits. We show that the expected total number of bits that must be revealed to implement this method grows at least linearly and at most quadratically in $n$ in the worst case.

Closing the gap between these lower and upper bounds remains an open problem. In particular, we conjecture that both the lower and upper bounds can be improved. It would also be interesting to determine the exact value of $\mathcal{B}(n)$ for small values of $n$.

\section*{Acknowledgment}

The authors would like to thank Luc Devroye from McGill university for bringing this problem to their attention and suggesting to write a paper about it if results were found as well as for providing feedback and suggestions that improved the writing of this paper. The authors would also like to thank Jun Kai Liao from McGill university for actively participating in the early stages of this research project.

\end{document}